\documentclass[a4paper,11pt]{article}

\usepackage{amsthm}
\usepackage{amssymb}
\usepackage{graphicx, subfigure}

\newtheorem{theorem}{Theorem}[section]

\newtheorem{proposition}[theorem]{Proposition}
\newtheorem{lemma}[theorem]{Lemma}

\newtheorem{remark}[theorem]{Remark}
\newtheorem{conjecture}[theorem]{Conjecture}

\title{Optimally Guarding 2-Reflex Orthogonal Polyhedra by Reflex Edge Guards}

\author{Giovanni Viglietta\thanks{JAIST, 1-1 Asahidai, Nomi-shi, Ishikawa, Japan. {\tt johnny@jaist.ac.jp}.}}

\index{Viglietta, Giovanni}

\begin{document}
\thispagestyle{empty}
\maketitle 

\begin{abstract}
Let an \emph{orthogonal polyhedron} be the union of a finite set of boxes in $\mathbb R^3$ (i.e., cuboids with edges parallel to the coordinate axes), whose surface is a connected 2-manifold. We study the NP-complete problem of guarding a non-convex orthogonal polyhedron having reflex edges in just two directions (as opposed to three, in the general case) by placing the minimum number of edge guards on reflex edges only.

We show that $$\left\lfloor \frac{r-g}{2} \right\rfloor +1$$ reflex edge guards are sufficient, where $r$ is the number of reflex edges and $g$ is the polyhedron's genus. This bound is tight for $g=0$. We thereby generalize a classic planar Art Gallery theorem of O'Rourke, which states that the same upper bound holds for vertex guards in an orthogonal polygon with $r$ reflex vertices and $g$ holes.

Then we give a similar upper bound in terms of $m$, the total number of edges in the polyhedron. We prove that $$\left\lfloor \frac{m-4}{8} \right\rfloor +g$$ reflex edge guards are sufficient, whereas the previous best known bound was $\lfloor 11m/72+g/6\rfloor-1$ edge guards (not necessarily reflex).

We also consider the setting in which guards are \emph{open} (i.e., they are segments without the endpoints), proving that the same results hold even in this more challenging case.

Finally, we show how to compute guard locations matching the above bounds in $O(n \log n)$ time.
\end{abstract}

\section{Introduction}\label{s:1}

\subsection*{Background}
In discrete geometry, the \emph{Art Gallery} problem asks to place a (preferably small) number of \emph{guards} in a given enclosure, so that the guards collectively see the whole enclosure. Most classic results on the Art Gallery problem are surveyed in~\cite{art,shermer,urrutia2000}.

If the enclosure is a simple polygon with $n$ vertices and the guards are points (and a guard $v$ sees a point $p$ if and only if the straight line segment $vp$ does not intersect the exterior of the polygon), then $\lfloor n/3\rfloor$ guards are always sufficient, and there are polygons in which they are also necessary (see~\cite{chvatal}). If the polygon is \emph{orthogonal} (i.e., its edges meet at right angles), then the optimal number of point guards reduces to $\lfloor n/4\rfloor$. Furthermore, if the orthogonal polygon has $g$ holes, it can be guarded by
\begin{equation}
\left\lfloor \frac{n}{4}+\frac{g}{2} \right\rfloor
\label{eq1}
\end{equation}
point guards, as established by O'Rourke (see~\cite[Theorem~5.1]{art}). Because $n=2r-4g+4$, where $r$ is the number of \emph{reflex} vertices in the orthogonal polygon,\footnote{This equation is folklore, and it can be proved by cutting the polygon along $g$ edge-parallel segments to eliminate the holes, and then computing the sum of the $n+4g$ internal angles of the resulting simple orthogonal polygon, observing that a convex (resp.~reflex) angle gives a contribution of $\pi/2$ (resp.~$3\pi/2$) and the total sum is $\pi(n+4g-2)$.} we may express the same upper bound in terms of $r$ as
\begin{equation}
\left\lfloor \frac{r-g}{2} \right\rfloor+1.
\label{eq2}
\end{equation}
Even though O'Rourke does not mention this fact, a careful analysis of his method (a decomposition of the polygon into L-shaped pieces, see~\cite[Chapters~2.5,~2.6]{art}) reveals that, if $r>0$, all the guards can be chosen to lie on \emph{reflex} vertices.

Aside from general lower and upper bounds, it is known that the problem of minimizing the number of guards needed to cover a specific polygon is NP-complete if the guards are to be located on its vertices (even if the polygon is orthogonal, see~\cite{kats}). The problem is complete for the existential theory of the reals ($\exists\mathbb R$) if guards can be located anywhere inside the polygon (see~\cite{tillman}).

In recent years, the attention has shifted to 3-dimensional enclosures, and especially to polyhedra (i.e., bounded regions of $\mathbb R^3$ with connected and piecewise linear surface). Point guards are much less effective in this setting: there exist polyhedra with $n$ vertices where guards placed at every vertex do not see the whole interior, and where $\Omega(n^{3/2})$ non-vertex guards are required (refer to~\cite[Chapter~10.2.2]{art}). This motivates the study of more powerful guards: an \emph{edge guard} is a guard that has the extent of an entire edge; a point $p$ is visible to an edge guard $e$ if and only if there is a point of $e$ that sees $p$. Ideally, if a point guard represents a stationary sentinel, an edge guard models a ``patrolling'' one. Equivalently, in the paradigm where the polyhedron is a room that has to be fully illuminated, a point guard represents a light bulb, and an edge guard models a fluorescent tube.

Cano et al.\ proved that any polyhedron with $m$ edges can be guarded by at most $\lfloor 27m/32\rfloor$ edge guards (see~\cite{cano}). For \emph{orthogonal} polyhedra (i.e., polyhedra whose faces meet at right angles), Urrutia found in~\cite{urrutia2000} that $\lfloor m/6\rfloor$ edge guards are sufficient, and conjectured that $m/12 + O(1)$ are optimal for polyhedra of genus~0. Urrutia's upper bound was later improved by Benbernou et al.\ in~\cite{cccg}: an orthogonal polyhedron with $m$ edges and genus $g$ can be guarded by $\lfloor 11m/72+g/6\rfloor -1$ edge guards; if the polyhedron has $r$ reflex edges, then $\lfloor 7r/12\rfloor -g+1$ edge guards are sufficient. The same paper also contains the conjecture that $r/2 + O(1)$ edge guards always suffice.

Another contribution of~\cite{cccg} was to show that the same upper bounds hold if the edge guards are deprived of their two endpoints: i.e., they hold for \emph{open edge guards}. This weaker type of guard has been introduced by the author in~\cite{thesis} and has been studied also in~\cite{open2} and~\cite{open1}. For a more in-depth discussion on this topic and some motivations, refer to~\cite[Chapter~3.2]{thesis}.

\emph{Face guards} have been explored as well, first by Souvaine et al.\ in~\cite{faceguards1}, and then by the author in~\cite{faceguards2}. Yet another line of research concerns guarding \emph{terrains}: a terrain is a piecewise linear surface embedded in $\mathbb R^3$ that intersects any vertical line in exactly one point. In the context of terrains, vertex guards, edge guards, and face guards have been studied in~\cite{bose}, \cite{bose,everett}, and~\cite{terrain1,terrain2,terrain3,faceguards2}, respectively.

A leitmotif of the author's Ph.D. thesis~\cite{thesis} is that edge guards in 3-dimensional polyhedra represent the ``natural counterpart'' of vertex guards in 2-dimensional polygons, in that they seem to exhibit analogous properties in the Art Gallery problem and in similar visibility-related pursuit-evasion problems. For instance, while vertex guards are insufficient to guard even orthogonal polyhedra\footnote{For an example, see the top-right polyhedron in Figure~\ref{f:examples}: there are points in the central region that are not visible to any vertex.} and face guards are an unrealistic model for patrolling guards,\footnote{The author showed in~\cite{faceguards2} that replacing a face guard $F$ with a path lying on $F$ that guards the same region as $F$ may produce a path of quadratic complexity.} edge guards are a more reasonable choice and yield upper and lower bounds that resemble the ones for vertex guards in 2-dimensional polygons. For instance, we have the aforementioned upper bounds and the conjectures of Urrutia and Benbernou et al.\ for orthogonal polyhedra. Here we formulate stronger versions of both conjectures, incorporating also the genus of the polyhedron, and requiring that guards lie on \emph{reflex} edges only:

\begin{conjecture}\label{con4:1}
Any non-convex orthogonal polyhedron with $m$ edges and genus $g$ is guardable by at most
$$\left\lfloor \frac{m}{12}+\frac{g}{2} \right\rfloor$$
(open or closed) reflex edge guards.
\end{conjecture}

\begin{conjecture}\label{con4:2}
Any orthogonal polyhedron with $r>0$ reflex edges and genus $g$ is guardable by at most
$$\left\lfloor \frac {r-g} 2\right\rfloor +1$$
(open or closed) reflex edge guards.
\end{conjecture}

Observe that the upper bound expressed by Conjecture~\ref{con4:2} is exactly the same as~(\ref{eq2}), which was proved by O'Rourke for (reflex) vertex guards in orthogonal polygons, here generalized to orthogonal polyhedra. To justify Conjecture~\ref{con4:1}, note that extruding an orthogonal polygon with $n$ vertices and $g$ holes produces an orthogonal prism with $m=3n$ edges and genus $g$ (by ``orthogonal prism'' we mean an orthogonal polygon that is also a prism). Now, we know that the polygon can be guarded by a number of reflex vertex guards bounded by~(\ref{eq1}), which implies that the prism can be guarded by the same number of (open or closed) reflex edge guards. Plugging $n=m/3$ into~(\ref{eq1}), we obtain exactly the upper bound of Conjecture~\ref{con4:1}.

So, both conjectures hold at least for orthogonal prisms, and in this paper we seek to prove that they hold for larger classes of orthogonal polyhedra, as well. We remark that, unlike with polygons, the two upper bounds expressed by our conjectures are not equivalent. This is because the clear-cut equation $n=2r-4g+4$ that holds for orthogonal polygons becomes the pair of inequalities
\begin{equation}
\frac m6 +2g-2\ \leqslant\ r\ \leqslant\ \frac 56\,m-2g-12
\label{eq3}
\end{equation}
for orthogonal polyhedra, both of which are tight for every $g$, as proved in~\cite{cccg,thesis}. However, there is a partial dependance between the two conjectures: the left-hand side of~(\ref{eq3}) implies that, if Conjecture~\ref{con4:1} is true, then so is Conjecture~\ref{con4:2}.

For $g=0$, Conjectures~\ref{con4:1} and~\ref{con4:2} have matching lower bounds: these were given by Urrutia in~\cite{urrutia2000} and by Benbernou et al.\ in~\cite{cccg}, and additional lower-bound examples for Conjecture~\ref{con4:2} are found at the end of Section~\ref{s:3}. By contrast, for arbitrary genus, the situation is still unclear, and even the corresponding 2-dimensional problem  (i.e., optimally guarding orthogonal polygons with holes) is open, see~\cite{art,urrutia2000}.

\subsection*{Our contributions}

This paper focuses on 2-reflex orthogonal polyhedra, i.e., orthogonal polyhedra whose reflex edges lie in at most two different directions. This is a case of intermediate complexity between the 1-reflex case (i.e., orthogonal prisms) and the full 3-reflex case (i.e., general orthogonal polyhedra). In Figure~\ref{f:examples} we illustrate examples of these three classes: in particular, note that the bottom polyhedron has reflex edges in only two ``horizontal'' directions, and is therefore 2-reflex (it also has genus~1).

\begin{figure}[h!]
\centering
\includegraphics[width=0.95\linewidth]{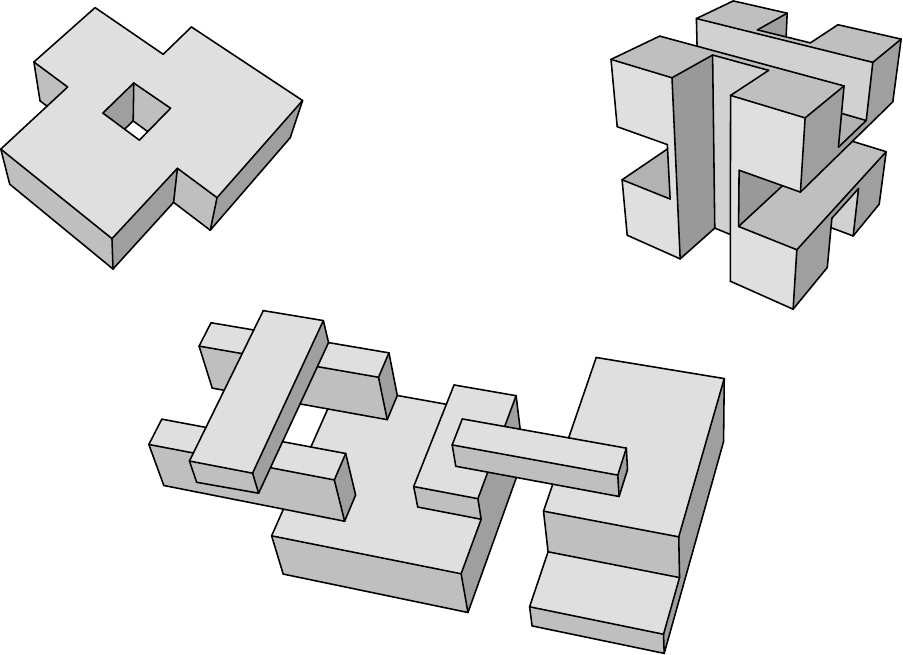}
\caption{Three orthogonal polyhedra. The top-left one is 1-reflex and the bottom one is 2-reflex. The top-right one is neither 1-reflex nor 2-reflex.}
\label{f:examples}
\end{figure}

2-reflex orthogonal polyhedra have been introduced by the author and investigated in the context of face guards in~\cite{faceguards2}. An equivalent but more practical way to define them is the following: take any finite set of boxes and ``attach'' them to each other by their bottom and top faces, but without letting their side faces touch. It should be apparent that 2-reflex orthogonal polyhedra can express a rich variety of 3-dimensional shapes, and are of interest in themselves. In Section~\ref{s:2} we formally analyze their structure and we establish some terminology.

There are also theoretical aspects to the study of 2-reflex orthogonal polyhedra: in line with our ultimate goal to prove Conjectures~\ref{con4:1} and~\ref{con4:2} for 3-reflex orthogonal polyhedra, and knowing that they hold at least for 1-reflex orthogonal polyhedra, the 2-reflex case seems a natural candidate of intermediate complexity. A proof of either conjecture limited to this important sub-case might be a necessary step toward a general proof, and we regard it as a key contribution: we believe that repeatedly ``cutting away'' 2-reflex orthogonal sub-polyhedra from a given 3-reflex orthogonal polyhedron eventually yields a ``core'' with enough structural properties to make it efficiently guardable.

In Section~\ref{s:3} we indeed prove Conjecture~\ref{con4:2} for 2-reflex orthogonal polyhedra. Along the way, we give auxiliary upper bounds for some sub-classes of polyhedra, and at the end of the section we show that they are all tight (for $g=0$).

In Section~\ref{s:4} we make progress toward Conjecture~\ref{con4:1} for 2-reflex orthogonal polyhedra, proving that
$$\left\lfloor \frac{m-4}{8} \right\rfloor +g$$
reflex edge guards are sufficient. This is an improvement upon the previous state of the art (limited to 2-reflex orthogonal polyhedra), in that it lowers the upper bound provided by Urrutia in~\cite{urrutia2000} by roughly $25\%$ and the one in~\cite{cccg} by roughly $18\%$ (for $g=0$). Also, it shows how guards can be chosen to lie on reflex edges, as opposed to arbitrary edges. Furthermore, we are able to prove Conjecture~\ref{con4:1} for a sub-class of 2-reflex orthogonal polyhedra called \emph{stacks} (for some examples of stacks see Figures~\ref{f:stack}--\ref{fig4:dcastle}, and for a definition refer to Section~\ref{s:2}).

In Section~\ref{s:5} we show that guard locations achieving the bounds given in Sections~\ref{s:3} and~\ref{s:4} can be computed in time $O(n\log n)$, where $n$ is the size of a representation of the polyhedron. In summary, through Sections~\ref{s:3}--\ref{s:5}, we will prove the following:
\begin{theorem}\label{t:main}
Given a 2-reflex orthogonal polyhedron of genus $g$ with $m$ edges, of which $r>0$ are reflex, a guarding set of at most 
$$\min\left\{\ \left\lfloor \frac{r-g}{2} \right\rfloor +1,\ \left\lfloor \frac{m-4}{8} \right\rfloor +g\ \right\}$$
(open or closed) reflex edge guards can be computed in $O(n \log n)$ time.
\end{theorem}

Finally, in Section~\ref{s:6} we discuss the computational complexity of finding the minimum number of guards for a given polyhedron. Among other results, we establish that the Art Gallery problem of guarding a 2-reflex orthogonal polyhedron with the minimum number of reflex edge guards is NP-complete. We remark that neither the NP-hardness of this problem nor its membership in NP is obvious.

A preliminary version of this paper has appeared in the author's Ph.D. thesis, see~\cite{thesis}.

\subsection*{Proof outline}

Here we outline our proof strategies for Sections~\ref{s:3} and~\ref{s:4}, which deal with Conjectures~\ref{con4:2} and~\ref{con4:1}, respectively. Due to the specific nature of our problem, all the known techniques for 3-dimensional Art Gallery problems are ineffective here. Namely, the techniques of Cano et al.\ in~\cite{cano} are designed for general polyhedra, and fail to take advantage of the structure of orthogonal polyhedra to place a small number of guards. Moreover, all the techniques in~\cite{cccg,cano,urrutia2000} place guards on arbitrary edges, whereas we insist on using only reflex edge guards.

Our proof in Section~\ref{s:3} follows the structure of O'Rourke's proof of the Art Gallery theorem for orthogonal polygons found in~\cite[Chapter~2.5]{art}. It turns out that the central part of O'Rourke's main argument can be considerably simplified, and then rephrased and generalized in terms of polyhedra (Lemmas~2.13--2.15 of~\cite{art} are condensed in our Lemma~\ref{l4:stack}), while the other parts of the proof require more sophisticated constructions and some novel ideas. We start by analyzing some sub-classes of polyhedra of increasing complexity: monotone orthogonal polyhedra, castles, double castles, and stacks. We then proceed to the main theorem.

For monotone orthogonal polyhedra we use a simple technique that is essentially the one used by O'Rourke for monotone orthogonal polygons. Castles and double castles require a deeper analysis based on some ad-hoc guarding strategies: these two classes of polyhedra roughly correspond to the ``histograms'' and ``double histograms'' in~\cite[Chapter~2.6]{art}. For stacks, we borrow from O'Rourke the idea of \emph{odd cuts}, which he used for 2-dimensional polygons, but also applies to this special class of polyhedra. An odd cut allows us to partition a stack into two smaller stacks and work on the two parts separately. After all possible odd cuts have been made, what is left is a collection of double castles, which we already know how to guard. However, this technique only applies to stacks of genus~0. For higher-genus stacks, we make some preliminary cuts to decrease the genus, and we show that our upper bound is still satisfied. For the full theorem, we consider a generic 2-reflex orthogonal polyhedron, and we show how to make some cuts to reduce its genus, and some further cuts to partition it into stacks. All these cuts are carefully made in such a way as to not compromise the upper bound we intend to prove.

In Section~\ref{s:4}, we seek an upper bound on reflex edge guards in terms of $m$ instead of $r$. We reuse the same construction as Section~\ref{s:3}, and we simply manipulate the bound we already have, replacing $r$ with an appropriate function of $m$. Ideally, we would like to use the right-hand-side inequality of~(\ref{eq3}), but it turns out to be too loose. Even specializing it for 2-reflex orthogonal polyhedra and sharpening it as much as possible fails to improve on the state of the art, and gives us essentially the upper bound of $\lfloor m/6\rfloor$ guards already found by Urrutia.

Our strategy is to take a step back and introduce a new parameter in the analysis of Section~\ref{s:3}. As a result, our Theorem~\ref{t4:2orthor} gives an upper bound in terms of $r$ and $g$ that also involves a third parameter $b$, which is the number of \emph{collars} in the polyhedron. A collar is a particular configuration of four reflex edges forming a rectangle, and introducing $b$ in our computations allows us to obtain the very sharp inequality $r\leqslant m/4 + 3g +b-3$. Plugging this expression into the one of Theorem~\ref{t4:2orthor} yields a bound of roughly $\lfloor m/8\rfloor$ guards, which improves on the state of the art (of course, limited to 2-reflex orthogonal polyhedra). Furthermore, in the case of stacks, we are able to prove that the left-hand-side of~(\ref{eq3}) holds with equality (as was the case with orthogonal polygons), implying that for stacks Conjecture~\ref{con4:1} is true.

As a possible direction for future research, we may try to improve the upper bound of $m/8+O(g)$ guards given in Section~\ref{s:4}, for example by treating separately other types of configurations, rather than just the collars.

\section{Definitions}\label{s:2}

For the purpose of this paper, an orthogonal polyhedron is defined as a connected 3-manifold (with boundary) in $\mathbb R^3$ that can be obtained as the union of finitely many \emph{boxes}, i.e., cuboids whose edges are parallel to the coordinate axes. Given an orthogonal polyhedron, two points $a$ and $b$ \emph{see} each other if and only if the segment $ab$ has no intersection with the exterior of the polyhedron. A point $p$ is \emph{guarded} by an edge $e$ of an orthogonal polyhedron if and only if there exists a point of $e$ that sees $p$. A subset $\mathcal S$ of the edges of an orthogonal polyhedron \emph{guards} the polyhedron itself if and only if every point of the polyhedron is guarded by at least one edge in $\mathcal S$.

We define the \emph{vertical} direction (or \emph{up-down} direction) to be the direction parallel to the $z$ axis. Accordingly, any direction parallel to the $xy$ plane is called \emph{horizontal}. Thus, the edges of an orthogonal polyhedron are either vertical or horizontal. We say that an orthogonal polyhedron is \emph{1-reflex} if it has no horizontal reflex edges, and \emph{2-reflex} if it has no vertical reflex edges (see Figure~\ref{f:examples}). Observe that the surface of a 2-reflex orthogonal polyhedron must be connected, and hence its \emph{genus} is well defined.

The intersection of a 2-reflex orthogonal polyhedron and a horizontal plane that does not contain any of its vertices is a collection of pairwise disjoint rectangles $R_1, R_2, \dots, R_k$ (note that this is not true of general orthogonal polyhedra, whose horizontal sections may be arbitrary orthogonal polygons). As the plane is moved upward or downward, the shape of the intersection does not change until a vertex is encountered. Focusing on one such rectangle $R_i$, the intersection between the polyhedron and the plane contains a copy of $R_i$ as long as the plane's ``altitude'' is within a certain interval: the union of such copies of $R_i$ is a cuboid. Thus, a natural way to partition a 2-reflex orthogonal polyhedron is into maximal cuboids, any two of which are either disjoint or touch each other at their top or bottom faces. Each cuboid in such a partition is called a \emph{brick}, and any non-empty intersection between two bricks is called their \emph{contact rectangle}. Figure~\ref{fig4:1} shows all the possible (combinatorial) ways in which two bricks can touch each other.

\begin{figure}[h!]
\centering
\subfigure[$m-m'=0$; $r-r'=4$]{\includegraphics[width=0.15\linewidth]{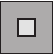}}\qquad
\subfigure[$m-m'=0$; $r-r'=3$]{\includegraphics[width=0.15\linewidth]{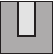}}\qquad
\subfigure[$m-m'=-3$; $r-r'=2$]{\includegraphics[width=0.15\linewidth]{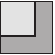}}\qquad
\subfigure[$m-m'=-6$; $r-r'=1$]{\includegraphics[width=0.15\linewidth]{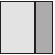}}\qquad
\subfigure[$m-m'=0$; $r-r'=2$]{\includegraphics[width=0.15\linewidth]{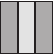}}\\
\subfigure[$m-m'=0$; $r-r'=4$]{\includegraphics[width=0.15\linewidth]{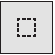}}\qquad
\subfigure[$m-m'=0$; $r-r'=3$]{\includegraphics[width=0.15\linewidth]{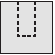}}\qquad
\subfigure[$m-m'=-3$; $r-r'=2$]{\includegraphics[width=0.15\linewidth]{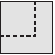}}\qquad
\subfigure[$m-m'=-6$; $r-r'=1$]{\includegraphics[width=0.15\linewidth]{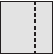}}\qquad
\subfigure[$m-m'=0$; $r-r'=2$]{\includegraphics[width=0.15\linewidth]{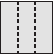}}\\
\subfigure[$m-m'=2$; $r-r'=3$]{\includegraphics[width=0.15\linewidth]{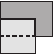}}\qquad
\subfigure[$m-m'=4$; $r-r'=4$]{\includegraphics[width=0.15\linewidth]{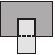}}\qquad
\subfigure[$m-m'=4$; $r-r'=3$]{\includegraphics[width=0.1504\linewidth]{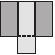}}\qquad
\subfigure[$m-m'=4$; $r-r'=4$]{\includegraphics[width=0.15\linewidth]{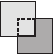}}\qquad
\subfigure[$m-m'=8$; $r-r'=4$]{\includegraphics[width=0.15\linewidth]{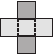}}\\
\subfigure[$m-m'=2$; $r-r'=3$]{\includegraphics[width=0.15\linewidth]{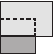}}\qquad
\subfigure[$m-m'=4$; $r-r'=4$]{\includegraphics[width=0.15\linewidth]{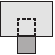}}\qquad
\subfigure[$m-m'=4$; $r-r'=3$]{\includegraphics[width=0.15\linewidth]{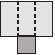}}\qquad
\subfigure[$m-m'=0$; $r-r'=2$]{\includegraphics[width=0.15\linewidth]{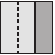}}\qquad
\subfigure[$m-m'=-1$; $r-r'=2$]{\label{f4:typet}\includegraphics[width=0.15\linewidth]{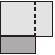}}
\caption{All possible contact rectangles of two adjacent bricks, viewed from above. The light-shaded bricks lie on top of the darker ones; reflex edges are shown in bold, and the dashed ones are those covered by the top brick (i.e., the ones invisible from above).}
\label{fig4:1}
\end{figure}

It is not hard to see that Figure~\ref{fig4:1} illustrates all such possibilities: it suffices to enumerate all the combinatorially different ways in which two distinct rectangles $R$ and $R'$ can overlap. We may do so by drawing $R$ and $R'$ on a $3\times 3$ grid and assuming that their intersection $S$ contains at least the central square of the grid. Note that $S$ is itself a grid-aligned rectangle consisting of either one, two, three, four, or six squares (it cannot contain all the nine squares of the grid, or $R$ and $R'$ would coincide). Suppose that $S$ contains only the central square, and observe that each edge of $S$ must lie either on the boundary of $R$ or on the boundary of $R'$. If all the edges of $S$ lie on the boundary of $R$ (resp.~$R'$), we have a configuration of type~(a) (resp.~(f)). If exactly three edges of $S$ lie on the boundary of $R$ (resp.~$R'$), we have a configuration of type~(l) (resp.~(q)). If two adjacent (resp.~opposite) edges of $S$ lie on the boundary of $R$ and the other two lie on the boundary of $R'$, we have a configuration of type~(n) (resp.~(o)). The cases where $S$ consists of more than one square of the grid can be enumerated in a similar way.

By ``cutting'' a 2-reflex orthogonal polyhedron along the contact rectangle of two adjacent bricks, all the reflex edges bordering the rectangle turn into convex edges, and the edge set will be modified according to the type of the contact rectangle. A convex edge may split in two distinct edges, new edges may be created, and multiple edges may merge together.

Figure~\ref{fig4:1} also indicates the number of edges and reflex edges gained or lost after such a split, for each different configuration. $m$ and $m'$ are the number of edges in the polyhedron before and after the split, respectively. Similarly, $r$ and $r'$ are the number of reflex edges before and after the split, respectively.

For example, in Figure~\ref{fig4:2}, a type-(t) contact rectangle between two bricks is illustrated, before and after the cut. The polyhedron before the cut has $m=23$ edges, of which $r=2$ are reflex. The two polyhedra after the cut have $m'=24$ edges in total, of which $r'=0$ are reflex (cf.\ Figure~\ref{f4:typet}).

\begin{figure}[h!]
\centering
\includegraphics[width=\linewidth]{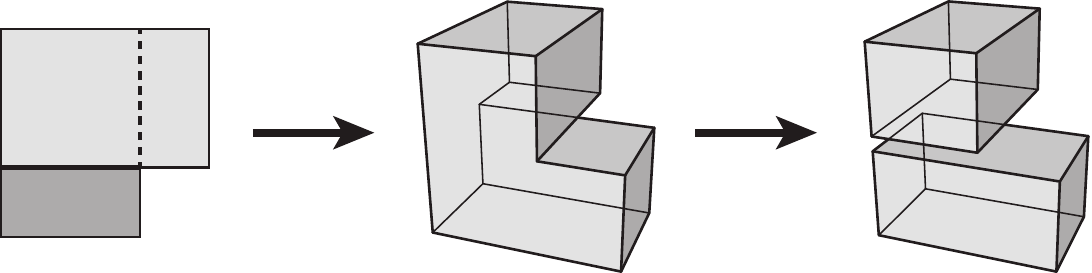}
\caption{Type-(t) contact rectangle between two bricks, before and after a cut (refer to Figure~\ref{fig4:1})}
\label{fig4:2}
\end{figure}

\begin{remark}\label{r4:degen}
When cutting a polyhedron of positive genus along a contact rectangle, we may fail to disconnect it, but just lower its genus. The resulting polyhedron is \emph{degenerate}, in that its boundary is self-intersecting: indeed, in place of the contact rectangle we now have a region that belongs to the top face of a brick and the bottom fare of another brick. We will occasionally encounter such degeneracies in intermediate steps of inductive proofs, and we will tolerate them as long as their presence does not invalidate our reasoning.
\end{remark}

Referring again to Figure~\ref{fig4:1}, we call each type-(a) or type-(f) contact rectangle a \emph{collar}, because its boundary is made of four reflex edges ``winding'' around a smaller brick. Singling out collars to treat them as separate cases will be useful in our proofs. The (technical) reason is that collars minimize the ratio $$\frac{m-m'+12}{r-r'}.$$ This ratio is 3 for collars, whereas it is at least 4 for any other contact type.

We also single out contact types~(d) and~(i), because they produce just one reflex edge each, and this will turn out to be the hardest case to handle. These two contact types will be called \emph{primitive}, and a 2-reflex orthogonal polyhedron having only primitive contact rectangles will be called a \emph{stack}. As Figure~\ref{f:stack} suggests, there are stacks of any genus.

\begin{figure}[h!]
\centering
\includegraphics[width=0.6\linewidth]{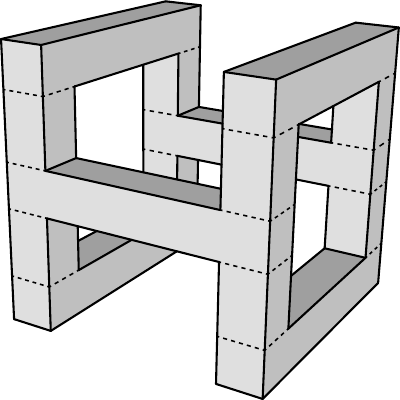}
\caption{Stack of genus 3 with dashed lines marking contact rectangles between bricks}
\label{f:stack}
\end{figure}

Observe that each brick of a stack may have up to two bricks attached to each of its horizontal faces. A stack where each brick has either zero or two other bricks attached to its top face is called a \emph{castle} (Figure~\ref{fig4:castle} shows an example). The bottom brick of a castle is called its \emph{base brick}. It is easy to see that a castle has an even number of reflex edges, because they all come in parallel pairs.

\begin{figure}[h!]
\centering
\includegraphics[width=0.65\linewidth]{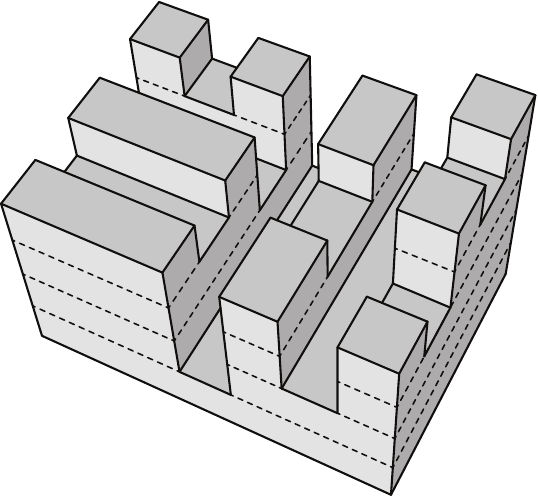}
\caption{Castle with dashed lines marking contact rectangles}
\label{fig4:castle}
\end{figure}

If a castle is turned upside down and its base brick is attached to another castle's base brick via a primitive contact rectangle, the resulting shape is a stack called \emph{double castle} (see Figure~\ref{fig4:dcastle}). It follows that a double castle has an odd number of reflex edges.

\begin{figure}[h!]
\centering
\includegraphics[width=0.5\linewidth]{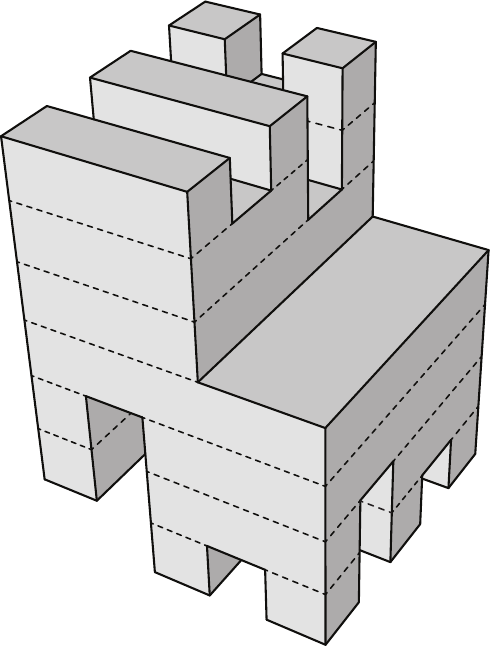}
\caption{Double castle with dashed lines marking contact rectangles}
\label{fig4:dcastle}
\end{figure}

Castles and double castles will play a fundamental role in the next section. Indeed, our castles are 3-dimensional cognates of the \emph{histograms} found in O'Rourke's proof (see~\cite[Chapter~2.6.1]{art}). Furthermore, it is easy to observe the similarity between our double castles and the ``double histogram'' in~\cite[Figure~2.59]{art}. Specifically, according to O'Rourke, a histogram is \emph{``an orthogonal polygon that has one horizontal edge (the base) equal in length to the sum of lengths of all the other horizontal edges''}. By analogy with histograms, we can observe that a castle has one horizontal face equal in area to the sum of areas of all the other horizontal faces. However, castles have additional structural properties, such as the fact that reflex edges come in parallel pairs. By extruding a histogram in general position (i.e., with no two collinear edges) where each reflex vertex is adjacent to another reflex vertex, we obtain a castle; however, the polyhedron in Figure~\ref{fig4:castle} demonstrates that not all castles can be constructed by extrusion. Conversely, any non-empty intersection between a castle and a vertical plane is a histogram, although not all histograms can be obtained by sectioning castles.

\section{Bounds in terms of $r$}\label{s:3}

In this section we are going to prove Conjecture~\ref{con4:2} for 2-reflex orthogonal polyhedra. As outlined in Section~\ref{s:1}, we first do so for some smaller classes of polyhedra, and then we work our way up to more general structures. At the end of the section, we will show that all the upper bounds given in our lemmas are tight.

We say that an orthogonal polyhedron is \emph{monotone} if it is a prism and if its intersection with any vertical line is either empty or a single line segment. Equivalently, a monotone orthogonal polyhedron can be constructed by extruding a monotone orthogonal polygon (refer to~\cite{art}). For instance, the prism in Figure~\ref{fig4:mono} is a monotone orthogonal polyhedron, because it is obtained by extruding its base (i.e., the light-shaded face), which is a monotone orthogonal polygon.

\begin{lemma}\label{l4:mono}
Any monotone orthogonal polyhedron with $r>0$ reflex edges is guardable by at most
\begin{equation}
\left\lfloor \frac{r}{2} \right\rfloor+1
\label{eq4:aaa}
\end{equation}
(open or closed) reflex edge guards.
\end{lemma}
\begin{proof}
Without loss of generality, suppose that the reflex edges are parallel to the $y$ axis (hence, all the points of a reflex edge have the same $x$ coordinate). Sort all reflex edges by increasing $x$ coordinate, breaking ties arbitrarily, and let $(e_i)_{1\leqslant i\leqslant r}$ be the sorted sequence. Now assign a guard to each edge whose index is an odd number, plus a guard to $e_r$ if $r$ is even. Thus, the bound (\ref{eq4:aaa}) is matched.

To show that the polyhedron is guarded, let $x_i$ be the $x$ coordinate of $e_i$, for $1\leqslant i\leqslant r$, and let $x_0$ and $x_{r+1}$ be the minimum and the maximum $x$ coordinate of a vertex of the polyhedron, respectively. Observe that the set of points whose $x$ coordinate lies in the interval $[x_{i}, x_{i+1}]$ is a cuboid, to which we refer as $\mathcal C_i$. Note that a guard lying on $e_i$ guards at least $\mathcal C_{i-1}$ and $\mathcal C_i$, since it lies on their common boundary, as Figure~\ref{fig4:mono} suggests. It follows that all the $\mathcal C_i$'s are guarded.
\end{proof}

\begin{figure}[h!]
\centering
\includegraphics[width=0.65\linewidth]{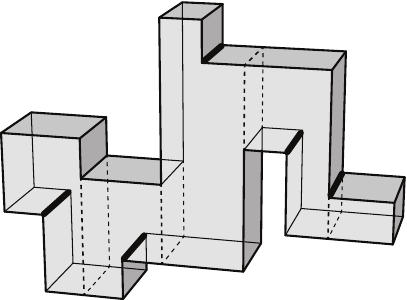}
\caption{Guarding a monotone orthogonal polyhedron. The reflex edges in bold represent guards.}
\label{fig4:mono}
\end{figure}

Observe that among the monotone orthogonal polyhedra there are all the castles that are also prisms. In the next lemma we show that we can actually use one less guard for the remaining castles.

\begin{lemma}\label{l4:castle}
Any castle with $2r$ reflex edges that is not a prism is guardable by at most $r$ (open or closed) reflex edge guards.
\end{lemma}
\begin{proof}
We prove our claim by well-founded induction on $r$. So, suppose the claim is true for all castles that are not prisms and have fewer than $2r$ reflex edges, and let $\mathcal C$ be a non-prism castle having exactly $2r$ reflex edges. Note that $\mathcal C$ cannot be a cuboid, and hence $r$ is strictly positive and two castles $\mathcal C_1$ and $\mathcal C_2$ must lie on top of the base brick of $\mathcal C$. Let $e_1$ (resp.\ $e_2$) be the reflex edge bordering the contact rectangle between the base brick of $\mathcal C$ and the base brick of $\mathcal C_1$ (resp.\ $\mathcal C_2$). Let $2r_1$ and $2r_2$ be the numbers of reflex edges of $\mathcal C_1$ and $\mathcal C_2$, respectively. It follows that $$r=r_1+r_2+1.$$
We have three cases.
\begin{itemize}
\item If neither $\mathcal C_1$ nor $\mathcal C_2$ is a prism, they both satisfy the inductive hypothesis and can be guarded by $r_1$ and $r_2$ reflex edge guards, respectively. Now we just assign a guard to $e_1$ in order to guard the base block of $\mathcal C$, and our upper bound of $r$ guards is matched.

\item If $\mathcal C_1$ is a prism and $\mathcal C_2$ is not (the symmetric case is analogous), then the induction hypothesis applies to $\mathcal C_2$, and we place $r_2$ reflex edge guards accordingly. Now, because $\mathcal C_1$ is a prism, its reflex edges are all parallel, and two sub-cases arise.
\begin{itemize}
\item If the reflex edges of $\mathcal C_1$ are parallel to $e_1$ (or $\mathcal C_1$ has no reflex edges), then $\mathcal C_1$ and the base brick of $\mathcal C$ together form a monotone orthogonal polyhedron with $2r_1+1$ reflex edges in total (see Figure~\ref{fig4:monocast}). By Lemma~\ref{l4:mono}, such a  polyhedron can be guarded by $r_1+1$ guards. Along with the previously assigned $r_2$ guards, this yields $r$ guards, as desired.

\begin{figure}[h!]
\centering
\includegraphics[width=0.7\linewidth]{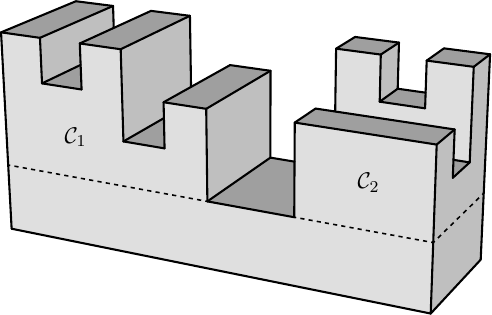}
\caption{The prism $\mathcal C_1$ and the base brick together form a monotone orthogonal polyhedron.}
\label{fig4:monocast}
\end{figure}

\item If the reflex edges of $\mathcal C_1$ are not parallel to $e_1$, then they are orthogonal to $e_1$. As a consequence, all of $\mathcal C_1$ is visible to $e_1$, and so is the base brick of $\mathcal C$, as Figure~\ref{fig4:monoguard} illustrates. It follows that $r_2+1$ guards are sufficient in this case, which is less than $r$.
\end{itemize}

\begin{figure}[h!]
\centering
\includegraphics[width=0.45\linewidth]{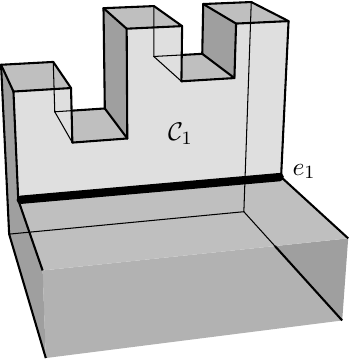}
\caption{Edge $e_1$ guards all of $\mathcal C_1$ plus the base brick.}
\label{fig4:monoguard}
\end{figure}

\item If both $\mathcal C_1$ and $\mathcal C_2$ are prisms, at least one of them (say, $\mathcal C_1$) must have reflex edges that are orthogonal to $e_1$, otherwise $\mathcal C$ itself would be a prism. Therefore $r_1\geqslant 1$, and $\mathcal C_1$ is guarded by assigning a guard to $e_1$, as in Figure~\ref{fig4:monoguard}. Again, two sub-cases arise.
\begin{itemize}
\item If the reflex edges of $\mathcal C_2$ are parallel to $e_2$, then $\mathcal C_2$ and the base brick of $\mathcal C$ form a monotone orthogonal polyhedron with $2r_2+1$ reflex edges (cf.~Figure~\ref{fig4:monocast}), which is guardable by $r_2+1$ reflex edge guards, by Lemma~\ref{l4:mono}. Overall, we assigned $1+r_2+1 \leqslant r_1+r_2+1=r$ guards, as required.

\item If the reflex edges of $\mathcal C_2$ are orthogonal to $e_2$, then $e_2$ guards $\mathcal C_2$, along with the base brick of $\mathcal C$ (refer to Figure~\ref{fig4:monoguard} again). We assigned only two guards, and $2\leqslant r_1+1\leqslant r$.
\end{itemize}
\end{itemize}
Hence, $r$ or fewer guards are sufficient in every case.
\end{proof}

The two previous lemmas enable us to prove that our upper bound holds at least for double castles.

\begin{lemma}\label{l4:dcastle}
Any double castle with $r$ reflex edges is guardable by at most
\begin{equation}
\left\lfloor \frac{r}{2} \right\rfloor +1
\label{eq4:ccc}
\end{equation}
(open or closed) reflex edge guards.
\end{lemma}
\begin{proof}
Let $\mathcal D$ be a double castle with $r$ reflex edges, made of a castle $\mathcal C_1$ and an upside-down castle $\mathcal C_2$ having $2r_1$ and $2r_2$ reflex edges, respectively. Let $e$ be the reflex edge lying on the contact rectangle between the two castles. Because $r=2r_1+2r_2+1$, (\ref{eq4:ccc}) can be rewritten as
\begin{equation}
r_1+r_2+1.
\label{eq4:ccc2}
\end{equation}

We distinguish three cases.
\begin{itemize}
\item If neither $\mathcal C_1$ nor $\mathcal C_2$ is a prism, by Lemma~\ref{l4:castle} they can be guarded by at most $r_1$ and $r_2$ reflex edge guards, respectively. Hence, (\ref{eq4:ccc2}) is matched.

\item If $\mathcal C_1$ is a prism and $\mathcal C_2$ is not (the symmetric case is analogous), by Lemma~\ref{l4:castle} $\mathcal C_2$ can be guarded by at most $r_2$ reflex edge guards. We have two sub-cases.
\begin{itemize}
\item If $r_1=0$, $\mathcal C_1$ is a cuboid and can be guarded by $e$. In total, $r_2+1$ guards are assigned, which matches (\ref{eq4:ccc2}).

\item If $r_1>0$, $\mathcal C_1$ can be guarded by $r_1+1$ guards, by Lemma~\ref{l4:mono}. Together with the previous $r_2$ guards, these match the upper bound (\ref{eq4:ccc2}).
\end{itemize}

\item If both $\mathcal C_1$ and $\mathcal C_2$ are prisms, we have three sub-cases.
\begin{itemize}
\item If the reflex edges of $\mathcal C_1$ and $\mathcal C_2$ are parallel to $e$, then $\mathcal D$ is a monotone orthogonal polyhedron with $2r_1+2r_2+1$ reflex edges, and according to Lemma~\ref{l4:mono} it is guardable by $r_1+r_2+1$ reflex edge guards, which agrees with (\ref{eq4:ccc2}). This holds also if $r_1=0$ or $r_2=0$.

\item If the reflex edges of $\mathcal C_1$ are orthogonal to $e$ and the reflex edges of $\mathcal C_2$ are parallel to $e$ (the symmetric case is analogous), then we may assume that $r_1\geqslant 1$ (otherwise we fall back to the previous case). We assign one guard to $e$ in order to guard $\mathcal C_1$ (cf.~Figure~\ref{fig4:monoguard}) and $r_2+1$ guards to reflex edges in $\mathcal C_2$, in accordance with Lemma~\ref{fig4:mono}. In total, we assigned $1+r_2+1 \leqslant r_1+r_2+1$ guards, thus matching (\ref{eq4:ccc2}).

\item Finally, if the reflex edges of both $\mathcal C_1$ and $\mathcal C_2$ are orthogonal to $e$, then a single guard assigned to $e$ sees all of $\mathcal D$ (refer to Figure~\ref{fig4:monodouble}), which obviously matches (\ref{eq4:ccc2}).
\end{itemize}
\end{itemize}
The bound is matched in all cases, which concludes the proof.
\end{proof}

\begin{figure}[h!]
\centering
\includegraphics[width=0.5\linewidth]{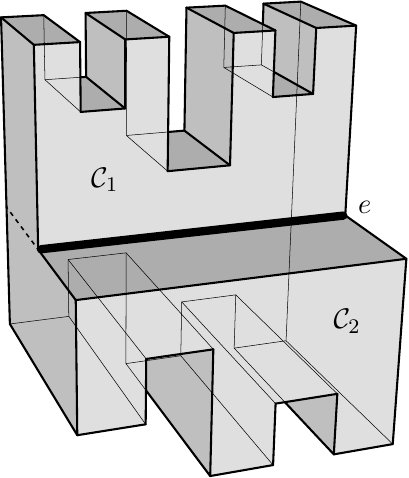}
\caption{Edge $e$ guards all of $\mathcal C_1$ and $\mathcal C_2$.}
\label{fig4:monodouble}
\end{figure}

Now that we know how to guard double castles, we can move on to more general shapes, such as stacks. This is the last step before our main theorem, and it generalizes~\cite[Lemmas~2.13--2.15]{art}.

Recall that a stack has only primitive contact rectangles. If two neighboring bricks $\mathcal B$ and $\mathcal B'$ share a type-(d) (resp.\ type-(i)) contact rectangle (refer to Figure~\ref{fig4:1}), and $\mathcal B$ lies below (resp.\ above) $\mathcal B'$, then we say that $\mathcal B$ is a \emph{parent} of $\mathcal B'$, and $\mathcal B'$ is a \emph{child} of $\mathcal B$. It follows that a brick in a stack can have at most two children above and two children below, and shares exactly one reflex edge with each child (see Figure~\ref{fig4:primitive}). Moreover, if a brick has one parent above (resp.\ below), then it has no other neighboring bricks above (resp.\ below). Finally, if a brick $\mathcal B$ has exactly one child $\mathcal B'$ on one side (regardless of the number of neighboring bricks on the other side), then the reflex edge shared by $\mathcal B$ and $\mathcal B'$ is said to be \emph{isolated}. It follows that the number of reflex edges in a stack and the number of isolated reflex edges have the same parity.

\begin{figure}[h!]
\centering
\subfigure[]{\includegraphics[width=0.25\linewidth]{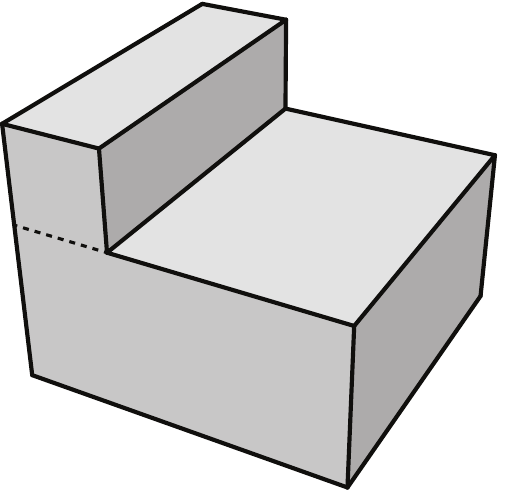}}\qquad
\subfigure[]{\includegraphics[width=0.25\linewidth]{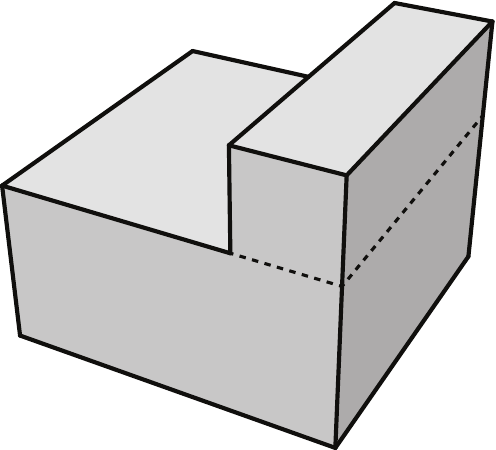}}\qquad
\subfigure[]{\includegraphics[width=0.25\linewidth]{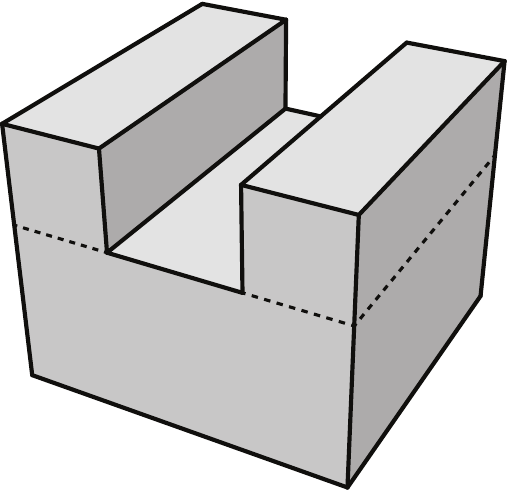}}\\ \vspace{0.25cm}
\subfigure[]{\includegraphics[width=0.25\linewidth]{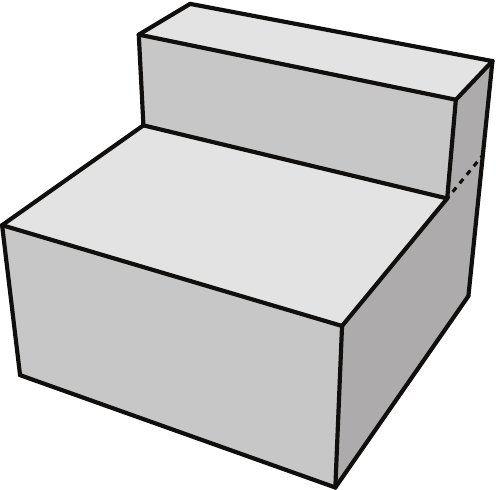}}\qquad
\subfigure[]{\includegraphics[width=0.25\linewidth]{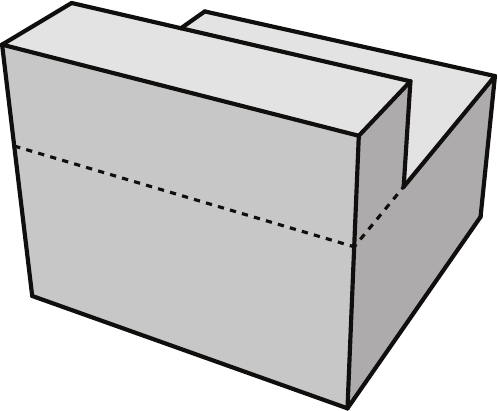}}\qquad
\subfigure[]{\includegraphics[width=0.25\linewidth]{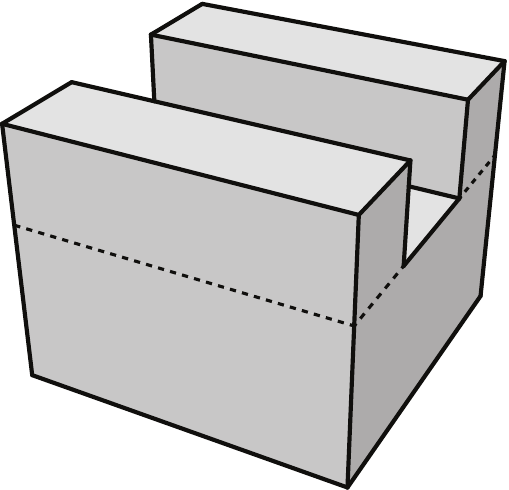}}
\caption{Different configurations for the upper children of a brick in a stack}
\label{fig4:primitive}
\end{figure}

\begin{lemma}\label{l4:stack}
Any stack with $r>0$ reflex edges and genus $g$ is guardable by at most
\begin{equation}
\left\lfloor \frac{r-g}{2} \right\rfloor+1
\label{eq4:ddd}
\end{equation}
(open or closed) reflex edge guards.
\end{lemma}
\begin{proof}
Our proof is by well-founded induction on $r$. Following O'Rourke (see~\cite[Chapter~2.5]{art}), we say that a contact rectangle $R$ in a simply connected stack $\mathcal S$ yields an \emph{odd cut} if $R$ partitions $\mathcal S$ into two stacks $\mathcal S_1$ and $\mathcal S_2$, one of which has an odd number of reflex edges. As an example, observe that every contact rectangle in a castle yields an odd cut (cf.~Figure~\ref{fig4:castle}). The presence of odd cuts in a simply connected stack is very desirable, in that it permits to successfully apply the inductive hypothesis.  Indeed, let $2r_1+1$ and $r_2$ be the number of reflex edges of $\mathcal S_1$ and $\mathcal S_2$ (the symmetric case is analogous). Then $r=2r_1+r_2+2$, because the cut resolves exactly one reflex edge. Recall that $\mathcal S$ is simply connected, and so are $\mathcal S_1$ and $\mathcal S_2$. Since $\mathcal S_1$ is non-convex, we can apply the inductive hypothesis on it, guarding it with at most
$$\left\lfloor \frac{2r_1+1}{2} \right\rfloor+1=r_1+1$$
reflex edge guards. Similarly, if $\mathcal S_2$ is non-convex, we can guard it with at most
\begin{equation}
\left\lfloor \frac{r_2}{2} \right\rfloor+1
\label{eq4:ddd2}
\end{equation}
reflex edge guards. On the other hand, if $\mathcal S_2$ is a cuboid, we can guard it by assigning a guard to the reflex edge lying on $R$. In this case, $r_2=0$, and hence (\ref{eq4:ddd2}) still holds. As a result, we have guarded all of $\mathcal S$ with at most
$$r_1+1+\left\lfloor \frac{r_2}{2} \right\rfloor+1 = \left\lfloor \frac{2r_1+r_2+2}{2} \right\rfloor+1=\left\lfloor \frac{r}{2} \right\rfloor+1$$
reflex edge guards. We stress that it makes sense to talk about odd cuts only for stacks of genus 0 (this fact will be used later in the proof).

Now, let $\mathcal S$ be a stack with exactly $r>0$ reflex edges and genus $g$, and assume that the lemma's statement holds for all non-convex stacks with fewer than $r$ reflex edges. There are four cases to consider.

\begin{itemize}
\item Let $g>0$. Note that each cut along a contact rectangle either disconnects $\mathcal S$ or lowers its genus by 1, and each cut resolves exactly one contact rectangle without creating or modifying other contact rectangles. Because $\mathcal S$ is partitioned by contact rectangles into cuboids, whose genus is zero, it follows that there exists at least one contact rectangle $R$ such that cutting through $R$ yields a (degenerate, see Remark~\ref{r4:degen}) stack $\mathcal S'$ with $r'=r-1$ reflex edges and genus $g'=g-1$. Observe that $\mathcal S'$ is non-convex, because it is made of at least two bricks, and so we can apply the inductive hypothesis on it and guard it with at most
$$\left\lfloor \frac{r'-g'}{2} \right\rfloor+1=\left\lfloor \frac{r-g}{2} \right\rfloor+1$$
reflex edge guards.

\item If $g=0$ and $r>0$ is even, then any contact rectangle yields an odd cut. Indeed, cutting $\mathcal S$ along a contact rectangle resolves exactly one reflex edge and produces two stacks (because $\mathcal S$ is simply connected). The amounts of reflex edges in these two stacks must have opposite parity, because their sum must be odd. Hence, one of them is odd.

\item Let $g=0$, and let $\mathcal S$ have a brick $\mathcal B$ with exactly one neighboring brick above and exactly one below (Figure~\ref{fig4:oneone} sketches one possible configuration for $\mathcal B$). We show that one of the two contact rectangles bordering $\mathcal B$ yields an odd cut. Let $R_1$ be the upper contact rectangle and $R_2$ be the lower one. If $R_1$ does not yield an odd cut, the stack above $R_1$ has an even number of reflex edges. But then, the stack above $R_2$, which additionally includes $\mathcal B$ and the reflex edge belonging to $R_1$, has an odd number of reflex edges. It follows that either $R_1$ or $R_2$ yields an odd cut.

\begin{figure}[h!]
\centering
\includegraphics[width=0.35\linewidth]{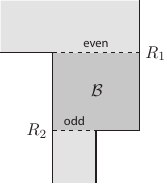}
\caption{Sketch of a brick with one neighbor above and one below}
\label{fig4:oneone}
\end{figure}

\item If none of the above are satisfied, we show that $\mathcal S$ must be a double castle, and hence our claim holds by Lemma~\ref{l4:dcastle}. Let then $r$ be odd, so that $\mathcal S$ has at least one isolated reflex edge corresponding to a contact rectangle $R$. Additionally, let no brick of $\mathcal S$ have exactly one neighbor above and one neighbor below. We show that $\mathcal S'$, i.e., the stack above $R$, is a castle. (By a symmetric argument, the stack below $R$ will be an upside-down castle, and $\mathcal S$ will then be a double castle.) Although this is easy to see (build the castle from bottom to top, without adding bricks that have only one neighbor below and one above), here follows a formal proof.

Let $d(\mathcal B)$ be the minimum number of bricks one has to traverse to reach brick $\mathcal B$ from $R$ (while always staying inside $\mathcal S'$), and let $\mathcal S'_h$ be the set of bricks $\mathcal B$ in $\mathcal S'$ such that $d(\mathcal B)\leqslant h$. We prove by induction on $h$ that $\mathcal S'_h$ is a castle whose base brick contains $R$.

The claim is true for $h=0$, because the brick just above $R$ (call it $\widetilde{\mathcal B}$) is the only one in contact with $R$, and a brick is indeed a castle.

Observe that no brick in $\mathcal S'$ can be attached to the bottom face of $\widetilde{\mathcal B}$, because the reflex edge of $\mathcal S$ corresponding to $R$ is isolated. Hence, $\widetilde{\mathcal B}$ must have either zero or two neighbors in $\mathcal S'$, and both of them are above. It follows that $\mathcal S'_1$ is a castle, as well.

Let now $h\geqslant 1$, assume that $\mathcal S'_h$ is a castle whose base brick $\widetilde{\mathcal B}$ contains $R$, and let us show that the same holds for $\mathcal S'_{h+1}$. Any brick $\mathcal B$ such that $d(\mathcal B)=h+1$ must be attached to some brick $\mathcal B'$ of $\mathcal S'_h$ such that $d(\mathcal B')=h$ (see Figure~\ref{fig4:stackcastle}). It is straightforward to see that any such $\mathcal B'$ has one parent brick below and no bricks above, in $\mathcal S'_h$. Hence, $\mathcal B$ can only be attached on top of $\mathcal B'$. Because $\mathcal B'$ cannot have only one top and one bottom neighbor in $\mathcal S'$, it follows that $\mathcal B$ cannot be a parent of $\mathcal B'$, but a child. Any other brick $\mathcal B''$ that is attached on top of $\mathcal B'$ in $\mathcal S'$ must also belong to $\mathcal S'_{h+1}$, by definition. As a consequence, one such brick $\mathcal B''$ must indeed be in $\mathcal S'_{h+1}$, otherwise $\mathcal B'$ would once again have only one top and one bottom neighbor in $\mathcal S'$. No new brick is attached to $\widetilde{\mathcal B}$, which remains the base of $\mathcal S'_{h+1}$. Thus the base brick of $\mathcal S'_{h+1}$ contains $R$, proving our claim and concluding the inductive proof.

\begin{figure}[h!]
\centering
\includegraphics[width=0.4\linewidth]{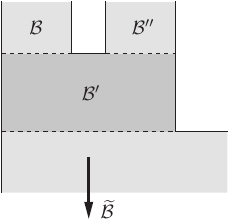}
\caption{Sketch of a brick in a castle}
\label{fig4:stackcastle}
\end{figure}

Since $\mathcal S'$ is connected and contains only a finite number of bricks, for a large-enough $h$ we have $\mathcal S'=\mathcal S'_h$, implying that $\mathcal S'$ is a castle.
\end{itemize}
Hence the lemma's statement holds for $\mathcal S$, which concludes the proof.
\end{proof}

In the next theorem, we give an upper bound on the number of guards for a 2-reflex orthogonal polyhedron in terms of $r$, $g$ and $b$, where $b$ is the number of collars. The presence of $b$ may look redundant (as it just contributes to lowering the bound), but we will actually have to carry this parameter along to the next section, in order to prove Theorem~\ref{t4:2orthoe}.

\begin{theorem}\label{t4:2orthor}
Any 2-reflex orthogonal polyhedron with $r>0$ reflex edges, $b$ collars and genus $g$ is guardable by at most
\begin{equation}
\left\lfloor \frac{r-g}{2} \right\rfloor -b +1
\label{eq4:main}
\end{equation}
(open or closed) reflex edge guards.
\end{theorem}
\begin{proof}
We proceed by induction on the number of non-primitive contact rectangles. The base case is given by non-convex stacks, for which (\ref{eq4:main}) holds due to Lemma~\ref{l4:stack} and the fact that $b=0$.

For the inductive step, let $\mathcal P$ be a 2-reflex orthogonal polyhedron with $r>0$ reflex edges, $b$ collars, genus $g$, and a non-primitive contact rectangle $R$. We cut $\mathcal P$ through $R$, thus resolving one non-primitive contact rectangle, and we distinguish two cases.
\begin{itemize}
\item If the cut does not disconnect $\mathcal P$, then it lowers its genus by 1. Let $\mathcal P'$ be the resulting polyhedron (which is degenerate, see Remark~\ref{r4:degen}), and let $g'=g-1$ be its genus. By inductive hypothesis, $\mathcal P'$ is guardable by
\begin{equation}
\left\lfloor \frac{r'-g'}{2} \right\rfloor -b'+1
\label{eq4:main2}
\end{equation}
guards, where $r'$ and $b'$ are, respectively, the number of reflex edges and collars of $\mathcal P'$.  Two sub-cases arise.
\begin{itemize}
\item If $R$ is a collar, then $r'=r-4$ and $b'=b-1$. Plugging these into (\ref{eq4:main2}), we obtain that $\mathcal P\setminus R$ is guardable by
$$\left\lfloor \frac{r-4-g+1}{2} \right\rfloor -b+1+1 \leqslant \left\lfloor \frac{r-g}{2} \right\rfloor -b+1$$
reflex edge guards.

\item If $R$ is not a collar, then $b'=b$. Because $R$ is not primitive, $r'\leqslant r-2$ (refer to Figure~\ref{fig4:1}). Hence, $\mathcal P$ is guardable by at most
$$\left\lfloor \frac{r-2-g+1}{2} \right\rfloor -b+1 \leqslant \left\lfloor \frac{r-g}{2} \right\rfloor -b+1$$
reflex edge guards.
\end{itemize}

\item If the cut disconnects $\mathcal P$ into $\mathcal P_1$ and $\mathcal P_2$, then $g_1+g_2=g$, where $g_1$ (resp.\ $g_2$) is the genus of $\mathcal P_1$ (resp.\ $\mathcal P_2$).  Let $r_1$ and $b_1$ (resp.\ $r_2$ and $b_2$) be, respectively,  the number of reflex edges and collars of $\mathcal P_1$ (resp.\ $\mathcal P_2$). Two sub-cases arise.

\begin{itemize}
\item If $R$ is a collar, then $r_1+r_2=r-4$ and $b_1+b_2=b-1$. If $r_1>0$, then we can apply the inductive hypothesis to $\mathcal P_1$ and guard it with at most
\begin{equation}
\left\lfloor \frac{r_1-g_1}{2} \right\rfloor -b_1+1
\label{eq4:main3}
\end{equation}
guards. Otherwise, $\mathcal P_1$ is a cuboid, and we can guard it by assigning a guard to any edge of $R$ (they are all reflex). In this case, $r_1=g_1=b_1=0$, and (\ref{eq4:main3}) is still matched. We proceed similarly with $\mathcal P_2$, and thus we have assigned a combined number of guards that is at most
$$\left\lfloor \frac{r_1-g_1}{2} \right\rfloor + \left\lfloor \frac{r_2-g_2}{2} \right\rfloor -b_1-b_2+2 \leqslant \left\lfloor \frac{r-4-g}{2} \right\rfloor -b+3 = \left\lfloor \frac{r-g}{2} \right\rfloor -b+1.$$
This many guards are then sufficient to guard $\mathcal P$.

\item If $R$ is not a collar, then $b_1+b_2=b$. Because $R$ is not primitive, $r_1+r_2\leqslant r-2$ (refer to Figure~\ref{fig4:1}). Once again, if $r_1>0$, then we can apply the inductive hypothesis to $\mathcal P_1$ and guard it with a number of guards that is bounded by (\ref{eq4:main3}). Otherwise, $\mathcal P_1$ is a cuboid, and we can guard it by assigning a guard to any reflex edge of $R$ (there is at least one). In this case, $r_1=g_1=b_1=0$, and (\ref{eq4:main3}) is still matched. Again, we do the same with $\mathcal P_2$, and thus we have assigned a combined number of guards that is at most
$$\left\lfloor \frac{r_1-g_1}{2} \right\rfloor + \left\lfloor \frac{r_2-g_2}{2} \right\rfloor -b_1-b_2+2 \leqslant \left\lfloor \frac{r-2-g}{2} \right\rfloor -b+2 = \left\lfloor \frac{r-g}{2} \right\rfloor -b+1,$$
and $\mathcal P$ is guarded.
\end{itemize}
\end{itemize}
We remark that also the rectangle $R$ is guarded in every case, because it is contained in $\mathcal P'$, $\mathcal P_1$, and $\mathcal P_2$, as defined above.
\end{proof}

Note that Theorem~\ref{t4:2orthor} implies that Conjecture~\ref{con4:2} holds for 2-reflex orthogonal polyhedra.

\subsection*{Lower bounds}

For $g=0$, the upper bound given in Theorem~\ref{t4:2orthor} is tight, as Figure~\ref{fig4:3} implies. For greater values of $g$ we have no matching lower bounds, and recall that the optimality problem is open even in the case of 2-dimensional polygons with holes.

\begin{figure}[h!]
\centering
\includegraphics[width=0.85\linewidth]{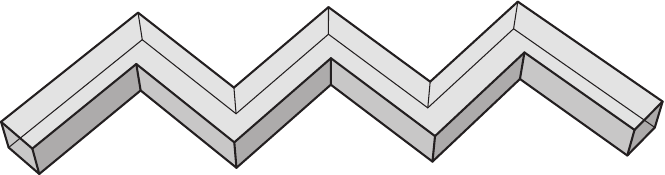}
\caption{2-reflex orthogonal polyhedron requiring $\lfloor r/2 \rfloor +1 = \lfloor m/12\rfloor$ reflex edge guards}
\label{fig4:3}
\end{figure}

Moreover, all the upper bounds given in Lemmas~\ref{l4:mono}--\ref{l4:stack} are tight, as well. Note that the polyhedron in Figure~\ref{fig4:3} is monotone and is also a stack, and therefore the bounds given by Lemmas~\ref{l4:mono} and~\ref{l4:stack} (for $g=0$) are matched. 

\begin{figure}[h!]
\centering
\includegraphics[width=0.5\linewidth]{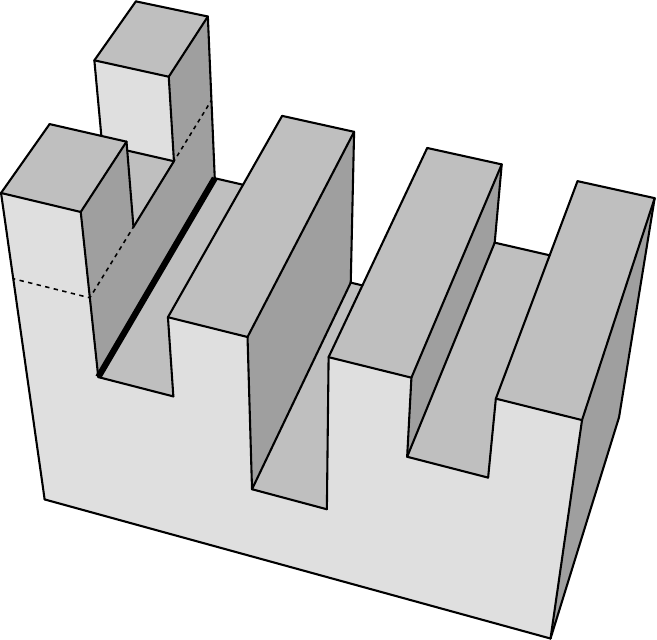}
\caption{Non-prism castle with $2r$ reflex edges that requires $r$ reflex edge guards. The edge in bold also guards the two top bricks.}
\label{fig4:4}
\end{figure}

The polyhedra in Figures~\ref{fig4:4} and~\ref{fig4:5} provide matching lower-bound examples for Lemmas~\ref{l4:castle} and~\ref{l4:dcastle}, respectively, and they are easy to generalize to arbitrarily large values of $r$.

\begin{figure}[h!]
\centering
\includegraphics[width=0.6\linewidth]{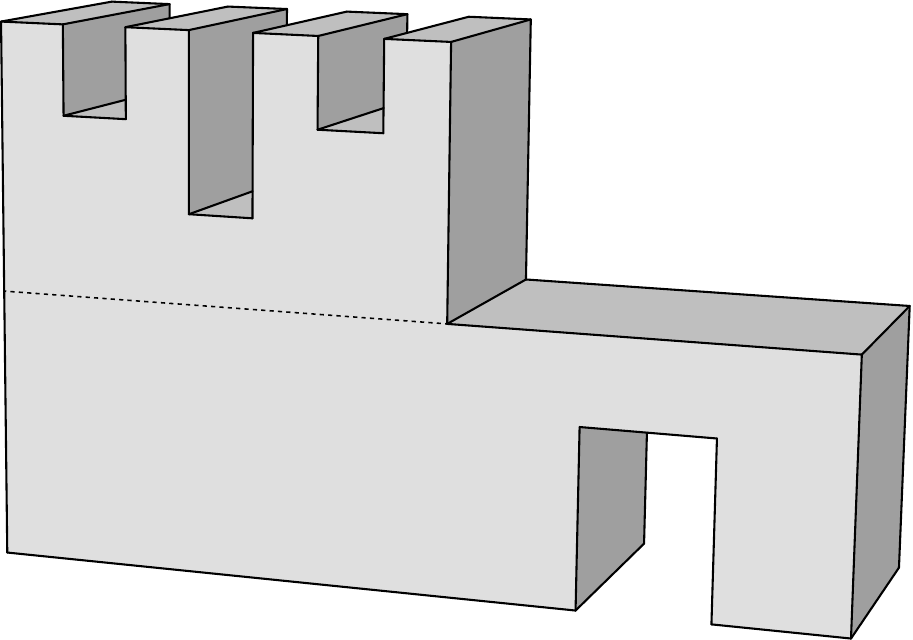}
\caption{Double castle requiring $\lfloor r/2 \rfloor +1$ reflex edge guards}
\label{fig4:5}
\end{figure}

\section{Bounds in terms of $m$}\label{s:4}

In this section we bound the number of reflex edge guards required to guard a 2-reflex orthogonal polyhedron in terms of $m$, as opposed to $r$. We also prove that Conjecture~\ref{con4:1} holds for all stacks. Rather than providing a radically new construction, we bound $m$ with respect to $r$, and then we apply Theorem~\ref{t4:2orthor}.

Note that a naive application of this method would not improve on the state of the art. Indeed, the sharpest possible inequality between $m$ and $r$ (involving also the genus $g$) is
\begin{equation}
m\geqslant 3r-12g+12,
\label{eq4:weakineq}
\end{equation}
which yields an upper bound of
$$\frac{m}{6}+O(g)$$
guards when applied to Theorem~\ref{t4:2orthor}. This result was already obtained by Urrutia in~\cite{urrutia2000} and improved in~\cite{cccg}.

To get around this, we will refine (\ref{eq4:weakineq}) by introducing the number of collars as an additional parameter.

\begin{lemma}\label{t4:ineq}
In every 2-reflex orthogonal polyhedron, the number of edges $m$, the number of reflex edges $r$, the number of collars $b$, and the genus $g$ satisfy the inequality
$$m\geqslant 4r-12g-4b+12.$$
\end{lemma}
\begin{proof}
We will prove that, for any collection of $k$ (pairwise internally disjoint) 2-reflex orthogonal polyhedra,
\begin{equation}
m\geqslant 4r-12g-4b+12k
\label{eq4:ineq}
\end{equation}
holds. Here, $m$ (resp.\ $r$, $g$, $b$) is the sum of the edges (resp.\ reflex edges, genera, collars) of the $k$ polyhedra. Then, by plugging $k=1$, we will obtain our claim.

Our proof proceeds by induction on $r$. For $r=0$ we have a collection of $k$ cuboids, each of which has 12 edges, so $m=12k$, $g=b=0$, and therefore (\ref{eq4:ineq}) holds as desired.

If $r>0$, there is at least one (horizontal) reflex edge, which is a side of the contact rectangle $R$ of two adjacent bricks, both belonging to the same polyhedron $\mathcal P$ of the collection. We can resolve this reflex edge (and up to three others) by separating the two bricks with a horizontal cut through $R$. As a consequence, either $\mathcal P$ gets partitioned in two polyhedra (in which case the new number of polyhedra is $k'=k+1$), or the genus of $\mathcal P$ decreases by 1 (in which case the new total genus is $g'=g-1$). Either way,
\begin{equation}
k'-g' = k-g+1.
\label{eq4:ineq4}
\end{equation}
By inductive hypothesis,
\begin{equation}
m'\geqslant 4r'-12g'-4b'+12k',
\label{eq4:ineq3}
\end{equation}
where $m'$ (resp.\ $r'$, $b'$) is the new number of edges (resp.\ reflex edges, collars), after the cut.

We have two cases. If $R$ is a collar, then $b'=b-1$, $m'=m$, and $r'=r-4$ (see Figure~\ref{fig4:1}). Plugging these into (\ref{eq4:ineq3}) and combining the result with (\ref{eq4:ineq4}) immediately yields (\ref{eq4:ineq}), as claimed.

Otherwise (i.e., $R$ is not a collar), by inspection of Figure~\ref{fig4:1}, it is clear that
\begin{equation}
m-m'\geqslant 4(r-r')-12
\label{eq4:ineq5}
\end{equation}
(recall that types (a) and (f) must be ignored, because they correspond to a collar). By combining (\ref{eq4:ineq4}), (\ref{eq4:ineq3}), (\ref{eq4:ineq5}) and plugging $b'=b$, we obtain again (\ref{eq4:ineq}), concluding the proof.
\end{proof}

\begin{theorem}\label{t4:2orthoe}
Any non-convex 2-reflex orthogonal polyhedron with $m$ edges and genus $g$ is guardable by at most
\begin{equation}
\left\lfloor \frac{m-4}{8} \right\rfloor +g
\label{e4:bound2}
\end{equation}
(open or closed) reflex edge guards.
\end{theorem}
\begin{proof}
Let $r>0$ be the number of reflex edges in the polyhedron. By Lemma~\ref{t4:ineq},
\begin{equation}
r\leqslant \frac m4+3g+b-3,
\label{eq4:ineq2}
\end{equation}
where $b$ is the number of collars. Applying (\ref{eq4:ineq2}) to Theorem~\ref{t4:2orthor}, we obtain that the number of guards is bounded by
$$\left\lfloor \frac {r-g}2 \right\rfloor -b+1 \leqslant \left\lfloor \frac m8 +\frac 32 g+\frac b2-\frac 32 - \frac g2 \right\rfloor -b+1 = \left\lfloor \frac{m-4}{8}-\frac b2 \right\rfloor +g\leqslant \left\lfloor \frac{m-4}{8} \right\rfloor +g.$$
\end{proof}

Recall that Conjecture~\ref{con4:1} holds at least for orthogonal prisms. Now, Lemma~\ref{l4:stack} implies that it holds more generally for stacks. Specifically, we have the following:

\begin{theorem}\label{t4:stackse}
Any non-convex stack with $m$ edges and genus $g$ is guardable by at most
\begin{equation}
\left\lfloor \frac{m}{12}+\frac{g}{2} \right\rfloor
\label{eq4:stackse}
\end{equation}
(open or closed) reflex edge guards.
\end{theorem}
\begin{proof}
A straightforward induction on the number of reflex edges $r$, based on Figure~\ref{fig4:1}, reveals that
$$m = 6r-12g+12$$
(recall that stacks only have type-(d) and type-(i) contact rectangles). Solving for $r$ and substituting it in (\ref{eq4:ddd}) immediately yields (\ref{eq4:stackse}).
\end{proof}

The upper bound of Theorem~\ref{t4:stackse} is tight for $g=0$, as the example in Figure~\ref{fig4:3} shows.

\section{Running time}\label{s:5}

In this section, we show how to efficiently compute guard positions matching the upper bounds given in Theorems~\ref{t4:2orthor} and~\ref{t4:2orthoe}. Notice that both bounds refer to the very same construction, described in Section~\ref{s:3}. In the present section, we will merely translate such a construction into an algorithm that runs in $O(n \log n)$ time.

Of course, to reason about running time, we have to specify how polyhedra are represented as data structures. We will assume that a polyhedron is stored as the array of its vertices, together with the array of its faces. Each face is a sequence of indices into the vertex array. The outer boundary of each face will be given in counterclockwise order with respect to a normal vector pointing outward, while its holes will be given in clockwise order.

We assume that numbers expressing vertex coordinates, as well as indices into the vertex array, have constant size. Note that, in an orthogonal polyhedron, each vertex is shared by at most six faces. So, if an orthogonal polyhedron has $n$ vertices, the size of its representation is $\Theta(n)$. Since the running time of our algorithm is given in big-oh notation, we do not have to make any distinction between the number of vertices and the size of the representation of a given orthogonal polyhedron.

O'Rourke's algorithm for \emph{simple} orthogonal polygons, detailed in~\cite[Chapter~2.6]{art}, also runs in $O(n \log n)$ time. As observed by Urrutia in~\cite{urrutia2000}, it could be optimized to run in $O(n)$ time, if Chazelle's linear-time triangulation algorithm were used (see~\cite{chazelletriang}). Unfortunately, Urrutia's speedup is only applicable to orthogonal polygons without holes.

In principle, we could rephrase O'Rourke's algorithm in the language of 2-reflex orthogonal polyhedra and obtain a new $O(n \log n)$ algorithm. However, four issues arise that require additional care:
\begin{enumerate}
\item\label{i4:1} O'Rourke's algorithm works on simply connected polygons, while our algorithm should be applied to polyhedra of any genus.

\item\label{i4:2} O'Rourke's algorithm may assign guards to convex vertices, whereas we insist on having guards only on reflex edges.

\item\label{i4:3} O'Rourke's method to find horizontal cuts in polygons does not trivially extend to polyhedra.

\item\label{i4:4} O'Rourke's algorithm relies on guarding double histograms, whereas we need guard double castles.
\end{enumerate}

We will now describe our modified algorithm, showing that each of the above issues has a relatively simple solution. Our algorithm takes as input a 2-reflex orthogonal polyhedron $\mathcal P$ with $n$ vertices, and outputs the set of reflex edges to which guards are assigned.

\subsection*{Preprocessing: $O(n)$}

First of all, we do some preprocessing on $\mathcal P$ to construct \emph{adjacency tables} of faces, edges and vertices, which allow us to efficiently navigate the polyhedron's surface. We mark each edge as reflex or convex and, if needed, we turn $\mathcal P$ by $90^\circ$ so that it contains no vertical reflex edges.

\subsection*{Finding contact lines: $O(n \log n)$}

We first compute a structure that is very similar to the \emph{horizontal visibility map} (also known as \emph{trapezoidization}, see~\cite{art}) of each vertical face of $\mathcal P$. This is a well-studied 2-dimensional problem that consists in partitioning a polygon into trapezoids by drawing horizontal lines at vertices. In our case, faces are orthogonal polygons (perhaps with holes), trapezoids are actually rectangles, and cut lines are drawn at reflex vertices only. For reasons that will be explained later, these cut lines are called \emph{contact lines}.

Note that some vertical faces of $\mathcal P$ may have ``degenerate'' vertices lying on horizontal edges. For instance, in the middle picture of Figure~\ref{fig4:2} there is a polyhedron made of two bricks, the top of which has a right-facing rectangular face $F$. However, five vertices of the polyhedron lie on the perimeter of $F$: the one lying in the middle of the bottom edge of $F$ is a degenerate vertex for $F$. In our algorithm, a degenerate vertex is treated as a pair of coincident vertices, one reflex and one convex: it is as if a degenerate vertex were actually a (degenerate) vertical edge of length 0.

Let $F_i$ be a vertical face of $\mathcal P$ with $n_i$ vertices. We sort all the vertical edges of $F_i$ by the $z$ coordinate of their upper vertex (in $O(n_i \log n_i)$ time), and we ``scan'' $F_i$ from top to bottom with a sweep line. We maintain a horizontally sorted list of all the vertical edges of $F_i$ pierced by the sweep line, in which insertion and deletion take $O(\log n_i)$ time. This list also includes the degenerate edges defined above. Every time our sweep line hits a vertex $v$ that is reflex in $F_i$ or belongs to a reflex edge of $\mathcal P$, we draw a contact line from $v$ to the next edge or the previous edge in the sorted list (according as the vertical edge containing $v$ is facing right or left), we add a \emph{dummy vertex} there (unless there is already a vertex), and we proceed with our sweep.

This process takes $O(n_i \log n_i)$ time and, letting $\widetilde n=\sum_i n_i$, finding the contact lines on every vertical face of $\mathcal P$ takes $O(n \log n)$ in total, because
$$\sum_i n_i \log n_i \leqslant \sum_i n_i \log \widetilde n = \widetilde n \log \widetilde n = \Theta(n\log n).$$

Every time we find a new contact line, we also update the face, edge, and vertex data we computed in the preprocessing step. That is, as soon as a new contact line is found, the corresponding face gets a new edge and is perhaps partitioned in two coplanar faces (this step takes constant time). Moreover, whenever we create a new dummy vertex $w$ for $F_i$, we add it to the other face sharing it (say, $F_j$). If we still have to process $F_j$, we treat $w$ as a reflex vertex and draw a contact line at $w$ in $F_j$ when we process it. Otherwise, $F_j$ is now a rectangle, and we just draw an additional contact line in it at $w$, in constant time: we call this procedure \emph{extension} of the contact line. However, if the extension generates yet another dummy vertex, there is no need to perform an ``extension of the extension'' to the next face, as we will see in a moment.

\subsection*{Finding bricks and contact rectangles: $O(n)$}

Notice that the contact lines constructed in the previous step are precisely the edges of the contact rectangles of $\mathcal P$. Indeed, we did not draw contact lines only at the reflex vertices of each face, but also at the dummy vertices created while processing other faces and at convex vertices that lie on reflex edges of $\mathcal P$.

Observe that a dummy vertex may be created only on a contact rectangle of type~(c), (d), (h), or~(i) (refer to Figure~\ref{fig4:1}). Indeed, in all other cases a vertex of a contact rectangle is also a vertex of the polyhedron. Moreover, only the contact rectangles of type~(d) and type~(i) have an edge with no endpoint on a vertex of $\mathcal P$: these are the only cases where the extension procedure is really needed, and performing it once per contact rectangle is enough to generate all its edges.

It is easy now to identify all the contact rectangles and all the bricks, navigating the boundary of $\mathcal P$ using the adjacency tables that we precomputed and then updated every time a cut line was drawn. While we identify the contact rectangles, we also build a \emph{brick graph} $G$, having a node for each brick and an edge connecting each pair of bricks sharing a contact rectangle.

Observe that issue~(\ref{i4:3}) above is now solved.

\subsection*{Resolving non-primitive contact rectangles: $O(n)$}

The non-primitive contact rectangles are those that are surrounded by more than one reflex edge of $\mathcal P$ and, as such, are easy to find. As proven in Theorem~\ref{t4:2orthor}, it is safe to cut $\mathcal P$ at a non-primitive contact rectangle and place guards in the resulting polyhedra (or polyhedron of lower genus).

Instead of actually cutting $\mathcal P$ and updating all the data structures, we merely delete the edges of $G$ corresponding to all non-primitive contact rectangles.

\subsection*{Ensuring simple connectedness: $O(n)$}

By this point, $\mathcal P$ has been partitioned into several, possibly not simply connected, stacks. As proven in Theorem~\ref{l4:stack}, it is safe to further cut the stacks until they all become simply connected. To do so, we again process only $G$, turning it into a forest. Such a task is accomplished by a depth-first traversal, starting at each connected component and deleting edges leading to already visited nodes. Recall that bricks in stacks have at most four neighbors, hence the time complexity of this traversal is indeed linear.

Observe that this step also solves issue~(\ref{i4:1}) above.

\subsection*{Adjusting brick parity: $O(n)$}

For reasons that will be clear shortly, we insist on having only stacks with an even number of bricks. O'Rourke achieves this by adding an extra ``chip'' to the polygon, in case he wants to change the parity of its reflex vertices. Then he applies his algorithm to the new polygon, and later removes the chip. If the chip happens to host a guard, then that guard is reassigned to the nearest convex vertex, after the chip is removed. Observe how this choice causes issue~(\ref{i4:2}) above.

In order to avoid placing guards on convex edges, we proceed as follows. We compute the size of each connected component of $G$ by a simple traversal. Then, in each component with an odd number of nodes, we find one leaf (recall that $G$ is a forest) and delete the edge attached to it (if the component is already an isolated node, we leave it as it is). Finally, we collect each isolated node, remove it from $G$, find its corresponding brick $\mathcal B$ in $\mathcal P$, find a contact rectangle bordering $\mathcal B$ (one must exist), find one reflex edge $e$ on it and assign it a guard. Referring to Figure~\ref{fig4:primitive}, it is obvious that $e$ guards all of $\mathcal B$.

The correctness of this step follows from the remarks contained in the proof of Lemma~\ref{l4:stack}, that every contact rectangle in a stack with an odd number of bricks yields an odd cut, and that it is always safe to make odd cuts.

\subsection*{Identifying odd cuts: $O(n)$}

We are left with stacks having an even number of bricks, and we want to further partition them with odd cuts. In order to identify odd cuts, we pick each connected component of $G$ and we do a depth-first traversal, rooted anywhere. During the traversal, we compute the parity of the size of the subtree dangling from each edge of $G$ we traverse. The parity is even if and only if that edge of $G$ corresponds to an odd cut.

Now, it is easy to observe that cutting a stack having an even number of bricks at an odd cut yields two stacks that again have an even number of bricks. Additionally, this operation does not change the parity of the cuts in the two resulting stacks. Hence, there is no need to re-identify the odd cuts after a cut is made. In contrast, stacks with an odd number of bricks do not have such a property, and this motivates our previous step.

It follows that we may remove all the edges of $G$ corresponding to odd cuts, without worrying about side effects.

\subsection*{Guarding double castles: $O(n)$}

At this point, only non-convex stacks without odd cuts are left. As a consequence of the observations in Lemma~\ref{l4:stack}, these are all double castles, which we now have to guard in linear time. In contrast, O'Rourke's algorithm was left at this point with double histograms (cf.~issue~(\ref{i4:4}) above).

Our algorithm is based on the proofs of Lemmas~\ref{l4:mono},~\ref{l4:castle}, and~\ref{l4:dcastle}, which naturally yield a procedure that cuts along certain contact rectangles and selects guards in monotone orthogonal polyhedra.

The only non-trivial aspect is that, occasionally in the procedure, we need know if some castle (or upside-down castle) is a prism, and what the orientation of its reflex edges is. To efficiently answer this question, we precompute this information for every ``sub-castle'' of each double castle that we have. We identify the two castles constituting each double castle (in linear time), then we do a depth-first traversal of the subgraphs of $G$ corresponding to those castles, starting from their base bricks. When we reach an internal node $v$, we recursively check if the subtrees dangling from its two children correspond to prisms, and if their reflex edges are oriented in the same direction. Then, after inspecting also the brick corresponding to $v$, we know if the subtree dangling from it corresponds to a prism and, if so, the direction of its reflex edges. Leaves are trivial to handle, in that they always correspond to prisms with no reflex edges.\\

Summarizing, and recalling the upper bounds given in Theorems~\ref{t4:2orthor} and~\ref{t4:2orthoe}, we have hereby proved Theorem~\ref{t:main}. Observe that the only superlinear step is the vertical sweep that finds the contact lines in every vertical face of the polyhedron. Whether a more efficient algorithm exists remains an open problem.

\section{Computational complexity}\label{s:6}

Aside from the general upper bounds on guard numbers discussed in this paper, one may wonder about the computational complexity of finding the minimum number of guards for any given polyhedron.

The 2-dimensional Art Gallery problem of guarding a simple orthogonal polygon by placing a minimum number of guards in it is shown to be NP-hard in~\cite{hardorthogonal}. The problem remains NP-hard with the additional restriction that guards may only be placed at the vertices of the polygon. Moreover, by slightly adapting a reduction from Vertex Cover in~\cite{kats}, we can show that the Art Gallery problem for simple orthogonal polygons is NP-hard even if guards can only be placed at reflex vertices.

These results immediately imply the NP-hardness of some 3-dimensional Art Gallery problems. Indeed, by extruding a simple orthogonal polygon with $r$ reflex vertices, we obtain a 1-reflex orthogonal polyhedron with genus~$0$ and $r$ reflex edges. Thus, placing a minimum number of point guards in such a polyhedron is NP-hard, as well as choosing a minimum number of reflex edge guards. Clearly, the NP-hardness of such problems extends to 2-reflex orthogonal polyhedra, as well as to more general classes of polyhedra.

Whether these problems are in NP is a subtler issue. In principle, guard coordinates could be used as a certificate, but it is not obvious that their representation would have polynomial size in general. In fact, it was proved in~\cite{tillman} that the Art Gallery problem for simple polygons is complete for the existential theory of the reals ($\exists\mathbb R$), provided that guards can be located anywhere in the polygon. This likely places such an Art Gallery problem outside of NP. By the extrusion argument, we immediately have that the Art Gallery problem for point guards in a polyhedron is $\exists\mathbb R$-hard, as well.

Moreover, a technique similar to the one in~\cite[Theorem~2.1]{tillman} allows us to prove that this 3-dimensional problem is also in $\exists\mathbb R$. The basic tool is an algebraic test for whether two points lie on the same side of a plane. If the plane has equation $ax+by+cz+d=0$ and the two points have coordinates $(x_1, y_1, z_1)$ and $(x_2, y_2, z_2)$ respectively, then they are on the same side if and only if $$(ax_1+by_1+cz_1+d)(ax_2+by_2+cz_2+d)>0.$$
With this formula, we can construct a polynomial-size system of inequalities to test if a point lies inside a given polyhedron, if two points can see each other, etc. The key observation is that the region of a polyhedron that is visible to a point guard is itself a polyhedron. Therefore, given a set of point guards, it is possible to subdivide the space into a polynomial number of convex regions, each of which contains a \emph{witness point}, as outlined in~\cite{tillman} for the 2-dimensional case. Such regions are obtained by drawing all the planes determined by the faces of the polyhedron, as well as all the planes identified by a point guard and an edge of the polyhedron. Each of these regions is either entirely visible or entirely invisible to any given guard. Computing a polynomial-size witness set containing at least one point in each such region is therefore possible within the existential theory of the reals. It is then easy to test if guard locations exist such that all the witness points that are contained in the polyhedron are also guarded.

\begin{proposition}
The Art Gallery problem for point guards in polyhedra is $\exists\mathbb R$-complete.
\end{proposition}

On the other hand, if the eligible guard locations are restricted to a finite set, such as the set of vertices or edges of a polygon or polyhedron, things are rather different. In a polygon, a set of vertices or edges is a certificate that can be verified in polynomial time. Since the region that is visible to a single vertex or edge is a polygon that can be computed in polynomial time, it is easy to test whether the union of some of these regions coincides with the whole polygon. As already observed above, the same is true of vertex guards in polyhedra: the region visible to a vertex is a polyhedron computable in polynomial time, and the union of two such polyhedra is easy to compute. Also, the NP-hardness of this problem, even restricted to 1-reflex orthogonal polyhedra, follows directly from its 2-dimensional version given in~\cite{hardorthogonal} and the extrusion technique.

\begin{proposition}
The Art Gallery problems for vertex guards in 1-reflex, 2-reflex, orthogonal, and general polyhedra are NP-complete.
\end{proposition}

For edge guards in polyhedra, the analysis gets more complicated. As observed in~\cite{faceguards2}, the region visible to an edge of an orthogonal polyhedron may not be a polyhedron: in general, its boundary is a piecewise quadric surface. Hence, in order to find witness points for a given set of edge guards as we did for point guards, we would have to compute intersections of quadric surfaces. Unfortunately, it is not clear whether this method can work in general, as the computations may yield coefficients containing radicals.

More specifically, the region visible to an edge guard in a polyhedron is bounded by the faces of the polyhedron itself plus some bundles of lines passing through the edge guard and one or two reflex edges. In the latter case, the three edges involved may determine the following surfaces:
\begin{itemize}
\item[(i)] a plane, if two of the three edges are parallel to each other;
\item[(ii)] a hyperbolic paraboloid, if the three edges are parallel to a common plane;
\item[(iii)] a hyperboloid of one sheet, otherwise.
\end{itemize}
Now, if the polyhedron is 2-reflex (or 1-reflex) and the guard itself coincides with a reflex edge, then (i) is always the case, because of three reflex edges two are necessarily parallel to each other. Therefore, in this special case, the region visible to a reflex edge guard is indeed a polyhedron, and witness points for a set of reflex edge guards are easy to compute. It follows that this Art Gallery problem is in NP, and hence NP-complete.

\begin{proposition}
The Art Gallery problem for reflex edge guards in 1-reflex and 2-reflex orthogonal polyhedra is NP-complete.
\end{proposition}

Determining the complexity of the Art Gallery problem for point guards in orthogonal polyhedra ($\exists\mathbb R$-complete?)\ and for (reflex) edge guards in orthogonal polyhedra (NP-complete?)\ remains open.

\paragraph{Acknowledgments.}
The author would like to thank the anonymous reviewers for helpful comments and suggestions, which considerably improved the readability of this paper.

\end{document}